\documentclass[a4paper,UKenglish]{oasics-v2018}
\usepackage{microtype}%if unwanted, comment out or use option "draft"
\usepackage{tcolorbox}
\usepackage{todonotes}
\usepackage[ruled,vlined]{algorithm2e}
\usepackage{afterpage}
\usepackage{placeins}

\SetKwFor{On}{on}{do}{end}

\nolinenumbers

\title{Submodular Optimization in the MapReduce Model}
\titlerunning{Submodular Optimization in the MapReduce Model} %optional, in case that the title is too long; the running title should fit into the top page column

\author{Paul Liu and Jan Vondrak}{Stanford University, USA}{\{paulliu, jvondrak\}@stanford.edu}{}{}

\authorrunning{P. Liu and J. Vondrak} %mandatory. First: Use abbreviated first/middle names. Second (only in severe cases): Use first author plus 'et. al.'

\Copyright{Paul Liu and Jan Vondrak}%mandatory, please use full first names. OASIcs license is "CC-BY";  http://creativecommons.org/licenses/by/3.0/

\subjclass{
Theory of computation $\rightarrow$ MapReduce algorithms; Distributed computing models; Algorithm design techniques; Submodular optimization and polymatroids
} 

\EventLongTitle{2nd Symposium on Simplicity in Algorithms (SOSA 2019)}
\EventShortTitle{SOSA 2019}
\EventAcronym{SOSA}
\EventYear{2019}
\EventLocation{San Diego, United States}
\EventLogo{}
\SeriesVolume{}
\ArticleNo{}
\begin{document}

\maketitle

\begin{abstract}
Submodular optimization has received significant attention in both practice and theory, as a wide array of problems in machine learning, auction theory, and combinatorial optimization have submodular structure. In practice, these problems often involve large amounts of data, and must be solved in a distributed way. One popular framework for running such distributed algorithms is MapReduce. In this paper, we present two simple algorithms for cardinality constrained submodular optimization in the MapReduce model: the first is a $(1/2-o(1))$-approximation in 2 MapReduce rounds, and the second is a $(1-1/e-\epsilon)$-approximation in $\frac{1+o(1)}{\epsilon}$ MapReduce rounds.
 \end{abstract}

\section{Introduction}
%auto-ignore

Let $f : 2^V \rightarrow \mathbb{R}^{+}$ be a function satisfying $f(A\cup\{e\}) - f(A) \geq f(B\cup\{e\}) - f(B)$ for all $A \subseteq B$ and $e \notin B$. Such a function is called \emph{submodular}. When $f$ satisfies the additional property $f(A\cup\{e\}) - f(A) \geq 0$ for all $A$ and $e \notin A$, we say $f$ is \emph{monotone}. 

Many combinatorial optimization problems can be cast as submodular optimization problems. Such problems include classics such as max cut, min cut, maximum coverage, and minimum spanning tree~\cite{F05}. Although submodular optimization encompasses several NP-Hard problems, well-known greedy approximation algorithms are known~\cite{NW78}. We focus on the special case of monotone submodular maximization under a cardinality constraint $k$, i.e. 
\[
OPT := \max_{S\subseteq V, |S|\leq k} f(S), \, \textrm{$f$ is monotone}.
\]
In particular, it is known that one can approximate a cardinality constrained monotone submodular maximization problem to a factor of $1-1/e$ of optimal.

Due to rapidly growing datasets, recent focus has been on submodular optimization in distributed models. In this work, we focus on the MapReduce model, where complexity is measured as the number of synchronous communication rounds between the machines involved. The current state of the art for cardinality constrained submodular maximization is the algorithm of Barbosa et al.~\cite{BENW16}, which achieves a $1/2-\epsilon$ approximation in 2 rounds and was the first to achieve a $1-1/e-\epsilon$ approximation in $O\left(\frac{1}{\epsilon}\right)$ rounds. Both algorithms actually require significant duplication of the ground set (each element being sent to $\Omega(\frac{1}{\epsilon})$ machines). Since this might be an issue in practice, \cite{BENW16} mentions that without duplication, the two algorithms could be implemented in $O(\frac{1}{\epsilon} \log \frac{1}{\epsilon})$ and $O(\frac{1}{\epsilon^2})$ rounds, respectively. Earlier, Mirrokni and Zadimoghaddam~\cite{MZ15} gave a $0.27$-approximation in 2 rounds without duplication and a $0.545$-approximation with $\Theta(\frac{1}{\epsilon} \log \frac{1}{\epsilon})$ duplication. 

\noindent{\bf Our contribution.}
We focus on the most practical regime of MapReduce algorithms for cardinality constrained submodular maximization, which is a small constant number of rounds and no duplication of the dataset. To our knowledge, the 0.27-approximation of \cite{MZ15} has been the best result in this regime so far.

We describe a simple thresholding algorithm which achieves the following:
In 2 rounds of MapReduce, with one random partitioning of the dataset (no duplication), we obtain a $(1/2-\epsilon)$-approximation. In 4 rounds, we obtain a $5/9$-approximation. More generally, in $2t$ rounds, we obtain a $(1 - (1-\frac{1}{t+1})^t - \epsilon)$-approximation, which we show to be optimal for this type of algorithm. Crucially, the parameter $\epsilon$ does not affect the number of rounds, and only mildly affects the memory (in that $\epsilon$ can be taken to $\tilde{O}(\sqrt{k/n})$ without asymptotically increasing the memory).

Our algorithm is inspired by the work of Kumar et al.~\cite{KMVV15} and McGregor-Vu~\cite{MV17} in the streaming setting. It is also similar to a recent algorithm of Assadi-Khanna~\cite{AK18}, who study the communication complexity of the maximum coverage problem. As such, our algorithm is not particularly novel, but we believe that our analysis of its performance in the MapReduce model is, thus simplifying and improving the previous work of \cite{BENW16} and \cite{MZ15}. 

\noindent{\bf Open question.}
The most intriguing remaining question in our opinion (for the cardinality constrained submodular problem) is whether $\Theta(1/\epsilon)$ rounds are necessary to achieve a $(1-1/e-\epsilon)$-approximation. So far there is no evidence that a $(1-1/e)$-approximation in a constant number of rounds is impossible.

\subsection{The MapReduce Model}
There are many variants of MapReduce models, and algorithms between the different models are largely transferable. We use a variant of the $\mathcal{MRC}$ model of Karloff et al.~\cite{KSV10}. In this model, an input of size $N$ is distributed across $O(N^\delta)$ machines, each with $O(N^{1-\delta})$ memory. We relax the model slightly, and allow one central machine to have memory slightly expanded to $\tilde{O}(N^{1-\delta})$.

Computation then proceeds in a sequence of synchronous communication rounds. In each round, each machine receives an input of size $O(N^{1-\delta})$. Each machine then performs computations on that input, and produces output messages which are delivered to other machines (specified in the message) as input at the start of the next round. The total size of these output messages must also be $O(N^{1-\delta})$ per machine. We refer the reader to the work of Karloff et al.~\cite{KSV10} for additional details.

In our applications, we assume the input is a set of elements $V$ and a cardinality parameter $k$. Each machine has an oracle that allows it to evaluate $f$. %For instance, $V$ could be a collection of sets, and $f$ could be the size of the union of sets that it is given. In this case, maximizing $f$ under a cardinality constraint $k$ is exactly the problem of maximum coverage with $k$ sets. 
Under these constraints, we assume that each machine has memory $O(\sqrt{nk})$ (except for a single `central' machine with $O(\sqrt{nk}\log k)$ memory) and that there are $\sqrt{n/k}$ machines in total.

\section{A thresholding algorithm for submodular maximization}
%auto-ignore

In the following algorithms, let $f : 2^V \rightarrow \mathbb{R}^+$ be a monotone submodular function, $n = |V|$, and $f_S(e) = f(S\cup\{e\}) - f(S)$. We refer to $f_S(e)$ as the \emph{marginal} of $e$ with respect to $S$. Let $k$ be the maximum cardinality of the solution, and $m = \sqrt{n/k}$ be the number of machines.

\subsection{A $1/2-o(1)$ approximation in 2 rounds}
\label{sec:1/2}

First, we present a simple 1/2-approximation in 2 rounds, assuming we know the exact value of $OPT$. We will relax this assumption later. The algorithm requires two helper functions {\sc ThresholdGreedy} and {\sc ThresholdFilter}, which forms the basis of all of our algorithms. Roughly speaking, {\sc ThresholdGreedy} greedily adds to a set of elements while there exists an element of high marginal in the input set. {\sc ThresholdFilter} filters elements of low marginal out of the input set. 

\begin{algorithm}[t]
  \DontPrintSemicolon
  \KwIn{
  An input set $S$, a partial greedy solution $G$ with $|G| \leq k$, and a threshold $\tau$.
  }
  \KwOut{
  A set $G^\prime \supseteq G$ such that $f_{G^\prime}(e) < \tau$ for all $e\in S$ if $|G|<k$ or $f(G)\geq \tau k$.
  }
  \caption{
  {\sc ThresholdGreedy}$(S, G, \tau)$
  }
  $G^\prime \leftarrow G$\;
  \For{$e \in S$}{
  \lIf{$f_{G^\prime}(e) \geq \tau$ \textbf{\upshape and} $|G^\prime| < k$}{
  $G^\prime \leftarrow G^\prime\cup\{e\}$
  }
  }
  \KwRet{$G^\prime$}
\label{alg:threshold-greedy}
\end{algorithm}

\begin{algorithm}[t]
  \DontPrintSemicolon
  \KwIn{
  An input set $S$, a partial greedy solution $G$, and a threshold $\tau$.
  }
  \KwOut{
  A set $S^\prime \subseteq S$ such that $f_{G}(e) \geq \tau$ for all $e\in S^\prime$.
  }
  \caption{
  {\sc ThresholdFilter}$(S, G, \tau)$
  }
  $S^\prime \leftarrow S$\;
  \For{$e \in S$}{
  \lIf{$f_{G}(e) < \tau$}{
  $S^\prime \leftarrow S^\prime\setminus\{e\}$
  }
  }
  \KwRet{$S^\prime$}
\label{alg:threshold-filter}
\end{algorithm}

We define an additional function {\sc PartitionAndSample} which simply initializes all of our algorithms by partitioning the input set randomly and drawing a random sample from it.
\begin{algorithm}[h]
  \DontPrintSemicolon
  %\KwIn{
  %An input set $S$.
  %}
  %\KwOut{
  %A set $S^\prime \subseteq S$ such that $f_{G}(e) \geq \tau$ for all $e\in S^\prime$.
  %}
  \caption{
  {\sc PartitionAndSample}$(V)$
  }
  $S \leftarrow \textrm{sample each $e\in V$ with probability $p = 4\sqrt{k/n}$}$\;
  partition $V$ randomly into sets $V_1, V_2, \ldots V_m$ to the $m$ machines (one set per machine)\; 
  send $S$ to each machine and a central machine $C$\;
\label{alg:partition-and-sample}
\end{algorithm}

Using these three helper algorithms, our approximation algorithm is quite easy to implement, and can be found in Algorithm~\ref{alg:one-half-two-round-with-guess}.

\begin{algorithm}[h]
  \DontPrintSemicolon
  \caption{A simple 2-round 1/2 approximation, assuming $OPT$ is known.}
\textbf{round 1:}\;
  $S,V_1,\ldots,V_m \leftarrow \textsc{PartitionAndSample}(V)$\footnotemark\;
  \On{each machine $M_i$ (in parallel)}{
    $\tau \leftarrow \frac{OPT}{2k}$\;
    $G_0 \leftarrow \textsc{ThresholdGreedy}\left(S, \emptyset, \tau\right)$\;
    \leIf{$|G_0|<k$}{
    $R_i\leftarrow \textsc{ThresholdFilter}\left(V_i, G_0, \tau\right)$\;
    }{
    $R_i \leftarrow \emptyset$
    }
    send $R_i$ to a central machine $C$
  }
\textbf{round 2 (only on $C$):}\;
  compute $G_0$ from $S$ as in first round\;
  $G \leftarrow \textsc{ThresholdGreedy}\left(\cup_i R_i, G_0, \tau\right)$\;
  \KwRet{$G$}
\label{alg:one-half-two-round-with-guess}
\end{algorithm}
\footnotetext{Note that $S$ and the $V_i$ are not stored on one machine by {\sc PartitionAndSample}. We simply use the assignment to denote that the variables have been initialized and sent to their respective machines.}

\begin{lemma}
\label{lem:half-apx}
The approximation ratio of Algorithm~\ref{alg:one-half-two-round-with-guess} is at least $1/2$.
\end{lemma}
\begin{proof}
The following lemma is folklore, but we present it for completeness.

First, we note that $G_0$ is the same on each machine so long as the loop iterating through $S$ is done in a fixed order. We assume that this is the case. From this, it is clear that Algorithm~\ref{alg:one-half-two-round-with-guess} returns a set $G$ for which $f_G(e) < \frac{OPT}{2k}$ for any $e\in V$.

Let $G$ be the set returned at the end of the algorithm. Either $|G|=k$, or there is no $e\in V$ for which the marginal with respect to $G$ is greater than $OPT/2$. In the former case, we have $k$ elements of value at least $\frac{OPT}{2k}$ so we are done. In the latter case, let $O$ be the optimal solution. By monotonicity and submodularity, 
\[
OPT = f(O) \leq f(O \cup G) \leq f(G) + \sum_{e \in O\setminus G} f_{G}(e) \leq f(G) + k\cdot\frac{OPT}{2k}. \qedhere
\]
\end{proof}

Lemma~\ref{lem:half-apx} shows that the algorithm is correct. Each machine in round 1 clearly uses $O(\sqrt{nk})$ memory. It remains to bound the memory of the central machine in round 2.

\begin{lemma}
\label{lem:memory-use-2-round-with-guess}
With probability $1 - e^{-\Omega(k)}$, the number of elements sent to the central machine $C$ has cardinality at most $\sqrt{nk}$.
\end{lemma}

\begin{proof}
The expected number of elements in $S$ is $4 \sqrt{nk}$. By a Chernoff bound (Theorem~\ref{thm:chernoff}) the probability that $|S| < 3 \sqrt{nk}$ is at most $e^{-\Omega(\sqrt{nk})} \leq e^{-\Omega(k)}$. So we can assume that $|S| \geq 3 \sqrt{nk}$.
Let $N_S$ denote the number of elements of marginal at least $OPT/(2k)$ with respect to $G_0$. % obtained by Algorithm~\ref{alg:threshold-greedy} on the random sample $S$. 
The number of elements sent to $C$ in round two is exactly $N_S + |S|$. 

Consider breaking the sample set $S$ into $3k$ blocks of size $\sqrt{n/k}$ and processing each block sequentially. If before each block, there are at least $\sqrt{nk}$ remaining elements of marginal value at least $OPT / (2k)$, we have probability at least $1 - \left(1 - \sqrt{\frac{k}{n}}\right)^{\sqrt{\frac{n}{k}}} > 1/2 $ of adding an additional element to $G_0$. This happens conditioned on any prior history of the algorithm, since we can imagine that the blocks are sampled independently one at a time. Therefore, we can use a martingale argument to bound the number of elements selected in $S$. If $X_i$ is the indicator random variable for the event that at least one element is selected from the $i$-th block, then we have $E[X_i \mid X_1,\ldots,X_{i-1}] \geq 1/2$. Hence we can define $Y_i = \sum_{j=1}^{i} (X_i - 1/2)$ and the sequence $Y_1,Y_2,\ldots$ is a submartingale, which means $E[Y_i \mid Y_1,\ldots,Y_{i-1}] \geq Y_{i-1}$. Moreover, $|Y_{i} - Y_{i-1}| \leq 1$. By Azuma's inequality (Theorem~\ref{thm:azuma}), $\Pr[Y_{3k} < -\frac12 k] < e^{-\Omega(k)}$. This means that with probability $1-e^{-\Omega(k)}$, $\sum_{j=1}^{3k} X_j = Y_{3k} + \frac{3}{2} k \geq k$, and we include at least $k$ elements overall. In that case, we are done and do not send anything to the central machine.
Otherwise, the number of remaining elements of marginal value at least $OPT / (2k)$ drops below $\sqrt{nk}$. 
\end{proof}

\noindent{\bf Remaining issues.}
Since we do not know the exact value of $OPT$, we will need to guess the value within a factor of $\epsilon$ without increasing the number of rounds. This will increase memory usage on the central machine by a factor of $\frac{1}{\epsilon}\log k$. To do this, we classify the inputs into two classes: when a the input contains more than  $\sqrt{nk}$ elements of value at least $\frac{OPT}{2k}$, and when there are less than $\sqrt{nk}$ such elements. We call the former class of inputs ``dense'' and the latter class ``sparse''. For each input class, we design a 1/2-approximation in 2 rounds. Given the input, we can run both in parallel and return the better of the two solutions:  each machine simply runs both algorithms at the same time, keeping the number of machines the same. The full analysis is given in the Appendix, but we outline the algorithms below.

\paragraph*{A 2-round algorithm for ``dense'' inputs}
Let $v$ be the maximum value of a single element of the random sample $S$ in Algorithm~\ref{alg:one-half-two-round-with-guess}. When the input is dense, $v$ is likely to be at least $\frac{OPT}{2k}$ and at most $OPT$. A straightforward analysis shows that $\tau_j := v(1+\epsilon)^j$ is within a $(1+\epsilon)$ multiplicative factor of $OPT/2$ for some $j\in \{1,\ldots,\frac{1}{\epsilon}\log k\}$. Running Algorithm~\ref{alg:one-half-two-round-with-guess} with $\tau_j$ instead of $OPT/2$ produces an approximation of value at least $\frac{OPT}{2(1+\epsilon)} > \frac{OPT}{2}(1-\epsilon)$. Thus if each machine runs $\frac{1}{\epsilon}\log k$ copies of Algorithm~\ref{alg:one-half-two-round-with-guess}, the best solution must have value at least $\frac{OPT}{2}(1-\epsilon)$.

\paragraph*{A 2-round algorithm for ``sparse'' inputs}
Call an element $e$ ``large'' if $f(e) \geq \frac{OPT}{2k}$. The algorithm simply sends all the large elements of the input onto one machine and then runs a sequential algorithm in the second round. To get all the large elements onto one machine, we randomly partition the input set onto the $m$ machines, and then send the $O(k)$ largest elements on each machine to the central machine. On the central machine, we can run the same thresholding procedure as in the ``dense'' case to find a threshold close to $OPT/(2k)$. We then run a sequential version of Algorithm~\ref{alg:one-half-two-round-with-guess}.

In both the algorithms, $\epsilon$ can be taken to $\tilde{O}(\sqrt{k/n})$ without asymptotically increasing the memory, so we have a $(1/2-o(1)$-approximation.
%auto-ignore

\subsection{A \texorpdfstring{$1-\left(1-\frac{1}{t+1}\right)^t$}{1-(1-\frac{1}{t+1})^t} approximation in \texorpdfstring{$2t$}{2t}  rounds}

Here we show how our algorithm extends to $t$ thresholds. The number of MapReduce rounds becomes $2t+2$. This can be reduced to $2t$ using tricks similar to Section~\ref{sec:1/2}, but we omit this here. The approximation factor with $t$ thresholds is $1 - (1 - \frac{1}{t+1})^t$, which converges to $1-1/e$. We note that we need $\Theta(1/\epsilon)$ rounds to obtain a $(1-1/e-\epsilon)$-approximation, similar to Barbosa et al.~\cite{BENW16}, but in contrast we do not need any duplication of the ground set. Barbosa et al. does not specify the constant factor in $\Theta(1/\epsilon)$ but it seems that our dependence is better; a calculation yields that we need $(1+o(1)) / \epsilon$ rounds to get a $(1-1/e-\epsilon)$-approximation.

For now, we assume (as in Algorithm~\ref{alg:one-half-two-round-with-guess}) that we know the exact value of $OPT$. We deal with this assumption later. In a nutshell our algorithm works just like Algorithm 1 but with multiple thresholds used in a sequence. We set the threshold values as follows:
$$ \alpha_\ell = \left( 1 - \frac{1}{t+1} \right)^\ell \frac{OPT}{k} $$
for $1 \leq \ell \leq t$. (Note that for $t=1$, we get $\alpha_1 = \frac{OPT}{2k}$ as in Algorithm 1.)
For each threshold, we first select elements above the threshold from a random sample set, and then use this partial solution to prune the remaining elements. Finally, the solution at this threshold is completed on a central machine, and we proceed to the next threshold. The full description of the algorithm is presented in Algorithm~\ref{alg:t-rounds}. The analysis is as follows.

\begin{lemma}
The approximation ratio of Algorithm~\ref{alg:t-rounds} is at least $1 - \left(1 - \frac{1}{t+1}\right)^t$.
\end{lemma}

\begin{proof}
By induction, we prove the following statement: The value of the first $\frac{\ell}{t} k$ elements selected by the algorithm is at least $(1 - (1 - \frac{1}{t+1})^\ell) OPT$. (If $\frac{\ell}{t} k$ is not an integer, we count the marginal value of the $\lceil \frac{\ell}{k} k \rceil$-th selected element weighted by its respective fraction.)

Clearly this is true for $\ell=0$. Assume that the claim is true for $\ell-1$. We consider two cases. 

Either all the elements among the first $\lceil \frac{\ell}{t} k \rceil$ are selected above the $\alpha_\ell$ threshold. This means that since the value of the first $\frac{\ell-1}{t} k$ elements was at least $(1 - (1-\frac{1}{t+1})^{\ell-1}) OPT$, and the marginal value of each additional element is at least $\alpha_\ell$, the total value of the first $\frac{\ell}{k} OPT$ (with fractional elements counted appropriately) is at least
$$ \left(1 - \left(1-\frac{1}{t+1}\right)^{\ell-1} \right) OPT + \frac{1}{t} \cdot \left(1 - \frac{1}{t+1}\right)^\ell OPT  = \left( 1 - \left(1 - \frac{1}{t+1} \right)^\ell \right) OPT.$$

The other case is that not all these elements are selected above the $\alpha_\ell$ threshold, which means that if we denote by $S_\ell$ the set of the first $\lfloor \frac{\ell}{t} k \rfloor$ selected elements, then there are no elements with marginal value more than $\alpha_\ell$ with respect to $S_\ell$. But then for the optimal solution $O$, we get
$$ OPT - f(S_\ell) \leq f_{S_\ell}(O) \leq k \alpha_\ell = \left( 1 - \frac{1}{t+1} \right)^\ell OPT $$
which means that $f(S_\ell) \geq (1 - (1 - \frac{1}{t+1})^\ell) OPT.$

For $\ell=t$, we obtain the statement of the lemma.
\end{proof}

The probabilistic analysis of the number of pruned elements that need to be sent to the central machine is exactly the same as in Section~\ref{sec:1/2}. The requirement of knowing $OPT$ can be also handled in the same way --- we can use an extra initial round to determine the maximum-value element on the input, which gives us an estimate of the optimum within a factor of $k$. Then we can try $O(\frac{1}{\epsilon}\log k)$ different estimates of $OPT$ to ensure that one of them is within a relative error of $1+\epsilon$ of the correct value. Finally, we use an extra final round to choose the best of the solutions that we found for different estimates of $OPT$. Alternatively, we can use additional tricks as in Section~\ref{sec:1/2} to eliminate these 2 extra rounds, but we omit the details here.
%at the cost of increasing the space requirements by a factor of $\tilde{O}(1/\epsilon)$ %paul: i think the different estimates of OPT already accounts for the space blowup.

\begin{algorithm}[h]
  \DontPrintSemicolon
  %\KwIn{
  %A monotone submodular function $f : 2^V \rightarrow \mathbb{R}^{+}$, an input set $V$, a cardinality parameter $k$, and $OPT$.
  %}
  %\KwOut{
  %A set $S \subseteq V$ such that $f(S) \geq \left(1-\left(\frac{t}{t+1}\right)^t\right)OPT$ and $|S|\leq k$.
  %}
  \caption{
  A $2t$-round $1-\left(\frac{t}{t+1}\right)^t$ approximation, assuming $OPT$ is known.
  }
  $G \leftarrow \emptyset$\;
  \For{$\ell=1,\ldots\,t$}{
    \textbf{round $2\ell-1$:}\;
    $S,V_1,\ldots,V_m \leftarrow \textsc{PartitionAndSample}(V)$\;
    \On{each machine $M_i$ (in parallel)}{
      $G_0 \leftarrow \textsc{ThresholdGreedy}\left(S, G, \alpha_\ell\right)$\;
      
      \leIf{$|G_0|<k$}{
      $R_i\leftarrow \textsc{ThresholdFilter}\left(V_i, G_0, \alpha_\ell\right)$\;
      }{
      $R_i \leftarrow \emptyset$
      }
      send $R_i$ to a central machine $C$
    }
  \textbf{round $2\ell$ (only on $C$):}\;
    compute $G_0$ from $S$ as in first round\;
    $G \leftarrow \textsc{ThresholdGreedy}\left(\cup_i R_i, G_0, \alpha_\ell\right)$\;
  }
  \KwRet{$G$}
\label{alg:t-rounds}
\end{algorithm}

%\begin{theorem}
%Algorithm~\ref{alg:t-rounds} achieves a $(1 - (1 - \frac{1}{t+1})^t)$-approximation in $2t+2$ rounds.
%\end{theorem}

\section{Optimality of our choice of thresholds}
Here we present a proof that there is no way to modify the thresholding algorithm and achieve a better approximation factor with a different choice of thresholds.

\begin{theorem}
The thresholding algorithm with $t$ thresholds cannot achieve a factor better than $1 - \left(1 - \frac{1}{t+1}\right)^t$.
\end{theorem}

\begin{proof}
Assume that the optimum $O$ consists of $k$ elements of total value $k v^*$. 
Since we are proving a hardness result, we can assume that the algorithm has this information and we can even let it choose $v^*$;
in the following, we denote this choice $v^* = \alpha_0$.
In addition, the algorithm chooses thresholds $\alpha_1 \geq \alpha_2 \geq \ldots \geq \alpha_t$.
It might be the case that $\alpha_0 < \alpha_1$, but then we can ignore all the thresholds above $\alpha_0$ and design our hard instance based on the thresholds below $\alpha_0$, which would reduce to a case with fewer thresholds. Thus we can assume $\alpha_0 \geq \alpha_1 \geq \ldots \geq \alpha_t$.

We design an adversarial instance as follows. In addition to the $k$ elements of value $v^*$, we have a set $S$ of other elements where element $i$ has value $v_i$, such that $\sum_{i \in S} v_i \leq k v^*$. The objective function is defined as follows: for $O' \subseteq O$ and $S' \subseteq S$,
$$ f(S' \cup O') = \sum_{i \in S'} v_i + \left(1 - \frac{\sum_{i \in S'} v_i}{k v^*} \right) |O'| v^*. $$
It is easy to verify that this is a monotone submodular function. (It can be realized as a coverage function, which we leave as an exercise.)

Now we specify more precisely the values of elements in $S$. We will have $n_\ell$ elements of value $\alpha_\ell$, for each $1 \leq \ell \leq t$. The idea is that the algorithm will pick these $n_\ell$ elements at threshold value $\alpha_\ell$, at which point the marginal value of the optimal elements drops below $\alpha_\ell$, so we have to move on to the next threshold. A computation yields that we should have $n_\ell = (\frac{\alpha_{\ell-1}}{\alpha_\ell} - 1) k$.\footnote{We ignore the issue that $n_\ell$ might not be an integer. For large $k$, it is easy to see that the rounding errors are negligible.} The total value of these elements is $\sum_{i \in S} v_i = \sum_{\ell=1}^{t} n_\ell \alpha_\ell = \sum_{\ell=1}^{t} (\alpha_\ell - \alpha_{\ell-1}) k = (\alpha_0 - \alpha_t) k \leq v^* k$ as required above.

Then, assuming that the marginal value of the optimum after processing $\ell-1$ thresholds was $\alpha_{\ell-1} k$, the marginal value after processing the $\ell$-th threshold will be $\alpha_{\ell-1} k - n_\ell \alpha_\ell = \alpha_\ell k$. By induction, the algorithm selects exactly $n_\ell$ elements of value $\alpha_\ell$, unless the constraint of $k$ selected elements is reached. Let us denote by $n'_\ell$ the actual number of elements selected by the algorithm at threshold level $\alpha_\ell$. We have $n'_\ell \leq n_\ell$, and $\sum_{\ell=1}^{t} n'_\ell \leq k$, as discussed above.

The total value collected by the algorithm is $\sum_{\ell=1}^{t} n'_\ell \alpha_\ell$. Since we have $n'_\ell \leq n_\ell = (\frac{\alpha_{\ell-1}}{\alpha_\ell} - 1) k$, and $\alpha_\ell \leq \alpha_{\ell-1}$, we can define inductively $\alpha'_0 = \alpha_0$ and $\alpha'_\ell \geq \alpha_\ell$ such that $n'_\ell = (\frac{\alpha'_{\ell-1}}{\alpha'_\ell} - 1) k$. Then the total value collected by the algorithm is 
$$\sum_{\ell=1}^{t} n'_\ell \alpha_\ell \leq \sum_{\ell=1}^{t} n'_\ell \alpha'_\ell
 = \sum_{\ell=1}^{t} (\alpha'_{\ell-1} - \alpha'_\ell) k = (\alpha'_0 - \alpha'_t) k. $$

Let us denote $\beta'_\ell = \frac{\alpha'_{\ell-1}}{\alpha'_\ell}$. We have $\sum_{\ell=1}^{t} (\beta'_\ell - 1) = \frac{1}{k} \sum_{\ell=1}^{t} n'_\ell \leq 1$, hence $\sum_{\ell=1}^{t} \beta'_\ell \leq t+1$. Recall that the value achieved by the algorithm is $\sum_{\ell=1}^{t} n'_\ell \alpha_\ell \leq (\alpha'_0 - \alpha'_t) k = (1 - 1 / \prod_{\ell=1}^{t} \beta'_\ell) v^* k$. By the AMGM inequality, $\prod_{\ell=1}^{t} \beta'_\ell$ is maximized subject to $\sum_{\ell=1}^{t} \beta'_\ell \leq t+1$ when all the $\beta'_\ell$ are equal, $\beta'_\ell = \frac{t+1}{t}$. Then, the value achieved by the algorithm is $(1 - 1 / \prod_{\ell=1}^{t} \beta'_\ell) v^* k = (1 - (\frac{t}{t+1})^t) OPT$.
\end{proof}

\bibliography{refs}

\begin{thebibliography}{1}

\bibitem{AK18}
Sepehr Assadi and Sanjeev Khanna.
\newblock Tight bounds on the round complexity of the distributed maximum
  coverage problem.
\newblock In {\em Proceedings of the Twenty-Ninth Annual {ACM-SIAM} Symposium
  on Discrete Algorithms, {SODA} 2018, New Orleans, LA, USA, January 7-10,
  2018}, pages 2412--2431, 2018.
\newblock URL: \url{https://doi.org/10.1137/1.9781611975031.155}, \href
  {http://dx.doi.org/10.1137/1.9781611975031.155}
  {\path{doi:10.1137/1.9781611975031.155}}.

\bibitem{BENW16}
Rafael da~Ponte~Barbosa, Alina Ene, Huy~L. Nguyen, and Justin Ward.
\newblock A new framework for distributed submodular maximization.
\newblock In {\em Proceedings of the IEEE 57th Annual Symposium on Foundations
  of Computer Science}, 2016.

\bibitem{F05}
Satoru Fujishige.
\newblock {\em Submodular functions and optimization}, volume~58.
\newblock Elsevier, 2005.

\bibitem{KSV10}
Howard~J. Karloff, Siddharth Suri, and Sergei Vassilvitskii.
\newblock A model of computation for {MapReduce}.
\newblock In {\em Proceedings of the Twenty-First Annual {ACM-SIAM} Symposium
  on Discrete Algorithms (SODA)}, pages 938--948, 2010.
\newblock \href {http://dx.doi.org/10.1137/1.9781611973075.76}
  {\path{doi:10.1137/1.9781611973075.76}}.

\bibitem{KMVV15}
Ravi Kumar, Benjamin Moseley, Sergei Vassilvitskii, and Andrea Vattani.
\newblock Fast greedy algorithms in {M}ap{R}educe and streaming.
\newblock {\em {TOPC}}, 2(3):14:1--14:22, 2015.
\newblock \href {http://dx.doi.org/10.1145/2809814}
  {\path{doi:10.1145/2809814}}.

\bibitem{MV17}
Andrew McGregor and Hoa~T. Vu.
\newblock Better streaming algorithms for the maximum coverage problem.
\newblock In {\em 20th International Conference on Database Theory, {ICDT}
  2017, March 21-24, 2017, Venice, Italy}, pages 22:1--22:18, 2017.
\newblock URL: \url{https://doi.org/10.4230/LIPIcs.ICDT.2017.22}, \href
  {http://dx.doi.org/10.4230/LIPIcs.ICDT.2017.22}
  {\path{doi:10.4230/LIPIcs.ICDT.2017.22}}.

\bibitem{MZ15}
Vahab~S. Mirrokni and Morteza Zadimoghaddam.
\newblock Randomized composable core-sets for distributed submodular
  maximization.
\newblock In {\em ACM Symposium on Theory of Computing (STOC)}, pages 153--162,
  2015.

\bibitem{NW78}
George~L. Nemhauser, Laurence~A. Wolsey, and Marshall~L. Fisher.
\newblock An analysis of approximations for maximizing submodular set functions
  - {I}.
\newblock {\em Math. Program.}, 14(1):265--294, 1978.
\newblock URL: \url{https://doi.org/10.1007/BF01588971}, \href
  {http://dx.doi.org/10.1007/BF01588971} {\path{doi:10.1007/BF01588971}}.

\bibitem{RS98}
Martin Raab and Angelika Steger.
\newblock ``balls into bins'' --- a simple and tight analysis.
\newblock In Michael Luby, Jos{\'e} D.~P. Rolim, and Maria Serna, editors, {\em
  Randomization and Approximation Techniques in Computer Science}, pages
  159--170, Berlin, Heidelberg, 1998. Springer Berlin Heidelberg.

\end{thebibliography}

\pagebreak

\appendix
\section{Appendix}
%auto-ignore

In Algorithm~\ref{alg:one-half-two-round-dense-class}, we design a 2-round $1/2-\epsilon$ approximation for ``dense'' inputs.
\begin{algorithm}[t]
  \DontPrintSemicolon
  \caption{
  A $1/2-\epsilon$ approximation in 2 rounds for ``dense'' inputs.
  }
  \textbf{round 1:}\;
  $S,V_1,\ldots,V_m \leftarrow \textsc{PartitionAndSample}(V)$\;
  \On{each machine $M_i$ (in parallel)}{
    $v \leftarrow \max_{e\textrm{ on $M_i$}} f(\{e\})$\;
    \For{$j=1,\ldots, \frac{1}{\epsilon}\log k$}{
      $\tau_j \leftarrow v(1+\epsilon)^j/k$\;
      $G_{0, j} \leftarrow \textsc{ThresholdGreedy}\left(S, \emptyset, \tau_j\right)$\;
	  
      \leIf{$|G_{0,j}|<k$}{
      $R_{i,j} \leftarrow \textsc{ThresholdFilter}\left(V_i, G_{0,j}, \tau_j\right)$\;
      }{
      $R_{i,j} \leftarrow \emptyset$
      }
    }
    send all $R_{i,j}$ to a central machine $C$
  }
\textbf{round 2 (only on $C$):}\;
  compute $v$, $\tau_j$, and $G_{0,j}$ from $S$ as in first round\;
  \For{$j=1,\ldots, \frac{1}{\epsilon}\log k$}{
    $G_j \leftarrow \textsc{ThresholdGreedy}\left(\cup_i R_{i,j}, G_{0,j}, \tau_j\right)$\;
  }
  $G^\star = \text{argmax}_{G_j} f(G_j)$\;
  \KwRet{$G^\star$}
\label{alg:one-half-two-round-dense-class}
\end{algorithm}

\begin{lemma}
\label{lem:close-guess}
Algorithm~\ref{alg:one-half-two-round-dense-class} returns a $1/2-\epsilon$ approximation.
\end{lemma}
\begin{proof}
Note that Algorithm~\ref{alg:one-half-two-round-dense-class} essentially runs Algorithm~\ref{alg:one-half-two-round-with-guess} with $O(\frac{1}{\epsilon}\log k)$ guesses for $OPT/2$. To get a $1/2-\epsilon$ approximation, one of these guesses needs to be within a multiplicative factor of $1+\epsilon$ from $\frac{OPT}{2k}$. By the denseness assumption on the input, we know $\frac{OPT}{2k}\leq v \leq OPT$ with high probability. Suppose we try $j=1,\ldots,\ell$ for some number of thresholds $\ell$. For one of the $\tau_j$'s to be within a multiplicative factor of $1+\epsilon$ from $\frac{OPT}{2k}$, we require $\ell \geq \log\left(\frac{OPT}{2kv\log(1+\epsilon)}\right) \geq \frac{1}{\epsilon}$.
\end{proof}

\begin{lemma}
The number of elements sent to the central machine is $O\left(\frac{1}{\epsilon}\sqrt{nk} \log k\right)$.
\end{lemma}
\begin{proof}
This follows from Lemma~\ref{lem:memory-use-2-round-with-guess} and the fact that there are only $\frac{\log k}{\epsilon}$ thresholds.
\end{proof}

Next, we design a 2-round $1/2-\epsilon$ approximation for ``sparse'' inputs (Algorithm~\ref{alg:one-half-two-round-sparse-class}).
\begin{algorithm}[t]
  \DontPrintSemicolon
  \caption{
  A $1/2-\epsilon$ approximation in 2 rounds for ``sparse'' inputs.
  }
  \textbf{round 1:}\;
  partition $V$ uniformly at random to the $m$ machines\;
  \On{each machine $M_i$}{
  	send its $O(k)$ largest elements to a central machine $C$\;
  }
  
  \textbf{round 2 (only on $C$):}\;
  $S \leftarrow \textrm{all elements sent to $C$}$\;
  $v \leftarrow \max_{e\in S} f(\{e\})$\;
  \For{$j=1,\ldots, \frac{1}{\epsilon}\log k$}{
  	$\tau_j \leftarrow v(1+\epsilon)^j/k$\;
  	$G_j \leftarrow \textsc{ThresholdGreedy}\left(S, \emptyset, \tau_j\right)$\;
  }
  $G^\star = \text{argmax}_{G_j} f(G_j)$\;
  \KwRet{$G^\star$}
\label{alg:one-half-two-round-sparse-class}
\end{algorithm}

%When the input set is ``sparse'', the algorithm is extremely simple. It shuffles all the large elements of the input onto one machine, and performs Algorithm~\ref{alg:one-half-two-round-with-guess} on it by guessing $\frac{\log k}{\epsilon}$ thresholds.

\begin{lemma}
Algorithm~\ref{alg:one-half-two-round-sparse-class} gives a $1/2-\epsilon$ approximation.
\end{lemma}
\begin{proof}
There are two things to check: that one of the $\tau_j$'s is close to $OPT/2$, and that the machine $C$ is not missing any of the elements that it needs. 

For the former, its clear that by similar reasoning to Lemma~\ref{lem:close-guess}, one of the $\tau_j$'s will be within a $1+\epsilon$ multiplicative factor to $OPT/2$. 

For the latter, note that the ``sparseness'' assumption implies that with high probability, the $\sqrt{nk}$ large elements will be equally distributed among the machines, and each machine will get $k$ elements in expectation. Since we send $O(k)$ elements to the central machine, $C$ will have all the large elements in $V$ with high probability. This can be shown by a standard balls-and-bins analysis~\cite{RS98}.
\end{proof}

Finally, we note that the total memory use on $C$ is $O(\sqrt{nk})$ since each machine sends $O(k)$ elements and there are $O(\sqrt{\frac{n}{k}})$ machines in total.

\begin{theorem}
By running Algorithms~\ref{alg:one-half-two-round-dense-class} and~\ref{alg:one-half-two-round-sparse-class} in parallel, we have a 2-round $1/2-\epsilon$ approximation.
\end{theorem}

\section{Auxiliary Results}
\begin{theorem}[Chernoff bound]
   \label{thm:chernoff}
   Let $X_1, \ldots, X_n$ be independent random variables such that $X_i \in [0,1]$ with probability 1. Define $X = \sum_{i=1}^n X_i$ and let $\mu = \mathbb{E} X$.
   Then, for any $\epsilon > 0$, we have
   \[
      \Pr[X \geq (1 + \epsilon) \mu ] \leq \exp\left( - \frac{\min\{\epsilon, \epsilon^2\} \mu}{3} \right).
   \]
\end{theorem}

\begin{theorem}[Azuma's inequality]
   \label{thm:azuma}
   Suppose $X_1, \ldots, X_n$ is a submartingale and $|X_i-X_{i+1} | \leq c_i$. Then, we have
   \[
      \Pr[X_n - X_0 \leq -t ] \leq \exp\left(\frac{-t^2}{2\sum_i c_i^2} \right).
   \]
\end{theorem}
%Some supplementary theorems.

\end{document}